\documentclass[9pt, conference]{IEEEtran}
\IEEEoverridecommandlockouts

\usepackage[T1]{fontenc}
\usepackage[LGR,T1]{fontenc}
\usepackage[utf8]{inputenc}
\usepackage{mathtools}

\usepackage{amssymb,mathrsfs}
\usepackage{amsthm}
\usepackage{bm}
\usepackage{scalerel}
\usepackage{nicefrac}
\usepackage{microtype} 
\usepackage[shortlabels]{enumitem}
\usepackage{graphicx}
\usepackage{epstopdf}
\DeclareGraphicsExtensions{.eps,.png,.jpg,.pdf}

\usepackage{url}
\usepackage{colortbl}
\usepackage{booktabs}
\usepackage{multirow}
\usepackage{colortbl,xcolor}
\usepackage{soul}
\usepackage{xparse,xstring}
\usepackage{calc}
\usepackage{etoolbox}

\makeatletter
\@ifpackageloaded{natbib}{
	\relax
}{
	\usepackage{cite}
}
\makeatother

\usepackage{array}
\newcolumntype{L}[1]{>{\raggedright\let\newline\\\arraybackslash\hspace{0pt}}m{#1}}
\newcolumntype{C}[1]{>{\centering\let\newline\\\arraybackslash\hspace{0pt}}m{#1}}
\newcolumntype{R}[1]{>{\raggedleft\let\newline\\\arraybackslash\hspace{0pt}}m{#1}}

\makeatletter
\let\MYcaption\@makecaption
\makeatother
\usepackage[font=footnotesize]{subcaption}
\makeatletter
\let\@makecaption\MYcaption
\makeatother

\usepackage{glossaries}
\makeatletter
\sfcode`\.1006

\let\oldgls\gls
\let\oldglspl\glspl

\newcommand\fussy@ifnextchar[3]{%
	\let\reserved@d=#1%
	\def\reserved@a{#2}%
	\def\reserved@b{#3}%
	\futurelet\@let@token\fussy@ifnch}
\def\fussy@ifnch{%
	\ifx\@let@token\reserved@d
		\let\reserved@c\reserved@a
	\else
		\let\reserved@c\reserved@b
	\fi
	\reserved@c}

\renewcommand{\gls}[1]{%
\oldgls{#1}\fussy@ifnextchar.{\@checkperiod}{\@}}
\renewcommand{\glspl}[1]{%
\oldglspl{#1}\fussy@ifnextchar.{\@checkperiod}{\@}}

\newcommand{\@checkperiod}[1]{%
	\ifnum\sfcode`\.=\spacefactor\else#1\fi
}

\robustify{\gls}
\robustify{\glspl}
\makeatother

\newacronym{wrt}{w.r.t.}{with respect to}
\newacronym{RHS}{R.H.S.}{right-hand side}
\newacronym{LHS}{L.H.S.}{left-hand side}
\newacronym{iid}{i.i.d.}{independent and identically distributed}
\newacronym{SOTA}{SOTA}{state-of-the-art}

\usepackage{float}

\ifx\notloadhyperref\undefined
	\ifx\loadbibentry\undefined
		\usepackage[hidelinks,hypertexnames=false]{hyperref}
	\else
		\usepackage{bibentry}
		\makeatletter\let\saved@bibitem\@bibitem\makeatother
		\usepackage[hidelinks,hypertexnames=false]{hyperref}
		\makeatletter\let\@bibitem\saved@bibitem\makeatother
	\fi
\else
	\ifx\loadbibentry\undefined
		\relax
	\else
		\usepackage{bibentry}
	\fi
\fi

\usepackage[capitalize]{cleveref-forward}

\crefname{equation}{}{}
\Crefname{equation}{}{}
\crefname{claim}{claim}{claims}
\crefname{step}{step}{steps}
\crefname{line}{line}{lines}
\crefname{condition}{condition}{conditions}
\crefname{dmath}{}{}
\crefname{dseries}{}{}
\crefname{dgroup}{}{}
\crefname{page}{page}{pages}

\crefname{Problem}{Problem}{Problems}
\crefformat{Problem}{Problem~#2#1#3}
\crefrangeformat{Problem}{Problems~#3#1#4 to~#5#2#6}

\crefname{Theorem}{Theorem}{Theorems}
\crefname{Corollary}{Corollary}{Corollaries}
\crefname{Proposition}{Proposition}{Propositions}
\crefname{Lemma}{Lemma}{Lemmas}
\crefname{Definition}{Definition}{Definitions}
\crefname{Example}{Example}{Examples}
\crefname{Assumption}{Assumption}{Assumptions}
\crefname{Remark}{Remark}{Remarks}
\crefname{Rem}{Remark}{Remarks}
\crefname{remarks}{Remarks}{Remarks}
\crefname{Appendix}{Appendix}{Appendices}
\crefname{Supplement}{Supplement}{Supplements}
\crefname{Exercise}{Exercise}{Exercises}
\crefname{TheoremA}{Theorem}{Theorems}
\crefname{CorollaryA}{Corollary}{Corollaries}
\crefname{PropositionA}{Proposition}{Propositions}
\crefname{LemmaA}{Lemma}{Lemmas}
\crefname{DefinitionA}{Definition}{Definitions}
\crefname{ExampleA}{Example}{Examples}
\crefname{RemarkA}{Remark}{Remarks}
\crefname{AssumptionA}{Assumption}{Assumptions}
\crefname{ExerciseA}{Exercise}{Exercises}
\crefname{algorithm}{Algorithm}{Algorithms}
\crefname{figure}{Fig.}{Figs.}
\crefname{table}{Table}{Tables}
\crefname{section}{Section}{Sections}
\crefname{subsection}{Section}{Sections}
\crefname{subsubsection}{Section}{Sections}

\usepackage{crossreftools}
\ifx\notloadhyperref\undefined
\else
	\relax
\fi

\usepackage{algorithm}
\usepackage{algpseudocode}

\ifx\loadbreqn\undefined
	\relax
\else
	\usepackage{breqn}
\fi

\makeatletter
\def\cleartheorem#1{%
    \expandafter\let\csname#1\endcsname\relax
    \expandafter\let\csname c@#1\endcsname\relax
}
\def\clearthms#1{ \@for\tname:=#1\do{\cleartheorem\tname} }
\makeatother

\ifx\renewtheorem\undefined
	\ifx\useTheoremCounter\undefined
		\newtheorem{Theorem}{Theorem}
		\newtheorem{Corollary}{Corollary}
		\newtheorem{Proposition}{Proposition}
		
	\else
		\newtheorem{Theorem}{Theorem}

	\fi

	\newtheorem{Assumption}{Assumption}

\fi

\theoremstyle{remark}

\theoremstyle{plain}

\newcommand{\qednew}{\nobreak \ifvmode \relax \else
		\ifdim\lastskip<1.5em \hskip-\lastskip
			\hskip1.5em plus0em minus0.5em \fi \nobreak
		\vrule height0.75em width0.5em depth0.25em\fi}



\newcommand{\nn}{\nonumber\\ }

\NewDocumentCommand{\movedownsub}{e{^_}}{%
	\IfNoValueTF{#1}{%
		\IfNoValueF{#2}{^{}}
	}{%
		^{#1}
	}%
	\IfNoValueF{#2}{_{#2}}
}

\let\latexchi\chi
\RenewDocumentCommand{\chi}{}{\latexchi\movedownsub}

\newcommand{\Real}{\mathbb{R}}

\newcommand{\calA}{\mathcal{A}}

\newcommand{\calE}{\mathcal{E}}

\newcommand{\calJ}{\mathcal{J}}
\newcommand{\calK}{\mathcal{K}}

\newcommand{\calP}{\mathcal{P}}

\newcommand{\calT}{\mathcal{T}}

\newcommand{\calV}{\mathcal{V}}

\newcommand{\calX}{\mathcal{X}}

\newcommand{\calZ}{\mathcal{Z}}

\DeclareSymbolFont{ttgreek}{LGR}{cmtt}{m}{n}
\DeclareMathSymbol{\ttalpha}{\mathord}{ttgreek}{`a}
\DeclareMathSymbol{\ttbeta}{\mathord}{ttgreek}{`b}
\DeclareMathSymbol{\ttgamma}{\mathord}{ttgreek}{`g}
\DeclareMathSymbol{\ttdelta}{\mathord}{ttgreek}{`d}
\DeclareMathSymbol{\ttepsilon}{\mathord}{ttgreek}{`e}
\DeclareMathSymbol{\ttzeta}{\mathord}{ttgreek}{`z}
\DeclareMathSymbol{\tteta}{\mathord}{ttgreek}{`h}
\DeclareMathSymbol{\tttheta}{\mathord}{ttgreek}{`j}
\DeclareMathSymbol{\ttiota}{\mathord}{ttgreek}{`i}
\DeclareMathSymbol{\ttkappa}{\mathord}{ttgreek}{`k}
\DeclareMathSymbol{\ttlambda}{\mathord}{ttgreek}{`l}
\DeclareMathSymbol{\ttmu}{\mathord}{ttgreek}{`m}
\DeclareMathSymbol{\ttnu}{\mathord}{ttgreek}{`n}
\DeclareMathSymbol{\ttxi}{\mathord}{ttgreek}{`x}
\DeclareMathSymbol{\ttomicron}{\mathord}{ttgreek}{`o}
\DeclareMathSymbol{\ttpi}{\mathord}{ttgreek}{`p}
\DeclareMathSymbol{\ttrho}{\mathord}{ttgreek}{`r}
\DeclareMathSymbol{\ttsigma}{\mathord}{ttgreek}{`s}
\DeclareMathSymbol{\tttau}{\mathord}{ttgreek}{`t}
\DeclareMathSymbol{\ttupsilon}{\mathord}{ttgreek}{`u}
\DeclareMathSymbol{\ttphi}{\mathord}{ttgreek}{`f}
\DeclareMathSymbol{\ttchi}{\mathord}{ttgreek}{`q}
\DeclareMathSymbol{\ttpsi}{\mathord}{ttgreek}{`y}
\DeclareMathSymbol{\ttomega}{\mathord}{ttgreek}{`w}

\DeclareMathSymbol{\ttAlpha}{\mathord}{ttgreek}{`A}
\DeclareMathSymbol{\ttBeta}{\mathord}{ttgreek}{`B}
\DeclareMathSymbol{\ttGamma}{\mathord}{ttgreek}{`G}
\DeclareMathSymbol{\ttDelta}{\mathord}{ttgreek}{`D}
\DeclareMathSymbol{\ttEpsilon}{\mathord}{ttgreek}{`E}
\DeclareMathSymbol{\ttZeta}{\mathord}{ttgreek}{`Z}
\DeclareMathSymbol{\ttEta}{\mathord}{ttgreek}{`H}
\DeclareMathSymbol{\ttTheta}{\mathord}{ttgreek}{`J}
\DeclareMathSymbol{\ttIota}{\mathord}{ttgreek}{`I}
\DeclareMathSymbol{\ttKappa}{\mathord}{ttgreek}{`K}
\DeclareMathSymbol{\ttLambda}{\mathord}{ttgreek}{`L}
\DeclareMathSymbol{\ttMu}{\mathord}{ttgreek}{`M}
\DeclareMathSymbol{\ttNu}{\mathord}{ttgreek}{`N}
\DeclareMathSymbol{\ttXi}{\mathord}{ttgreek}{`X}
\DeclareMathSymbol{\ttPi}{\mathord}{ttgreek}{`P}
\DeclareMathSymbol{\ttRho}{\mathord}{ttgreek}{`R}
\DeclareMathSymbol{\ttSigma}{\mathord}{ttgreek}{`S}
\DeclareMathSymbol{\ttTau}{\mathord}{ttgreek}{`T}
\DeclareMathSymbol{\ttUpsilon}{\mathord}{ttgreek}{`U}
\DeclareMathSymbol{\ttPhi}{\mathord}{ttgreek}{`F}
\DeclareMathSymbol{\ttChi}{\mathord}{ttgreek}{`Q}
\DeclareMathSymbol{\ttPsi}{\mathord}{ttgreek}{`Y}
\DeclareMathSymbol{\ttOmega}{\mathord}{ttgreek}{`W}

\newcommand{\sfJ}{\mathsf{J}}

\newcommand{\sfS}{\mathsf{S}}

\newcommand{\sfd}{\mathsf{d}}

\newcommand{\sfp}{\mathsf{p}}

\newcommand{\sfr}{\mathsf{r}}

\newcommand{\sft}{\mathsf{t}}

\newcommand{\sfv}{\mathsf{v}}

\newcommand{\bsfp}{\bm{\mathsf{p}}}

\newcommand{\bttXi}{\bm{\ttXi}}

\newcommand{\bh}{\mathbf{h}}

\newcommand{\bp}{\mathbf{p}}

\newcommand{\bs}{\mathbf{s}}
\newcommand{\bS}{\mathbf{S}}

\newcommand{\bbN}{\mathbb{N}}

\DeclareSymbolFont{bsfletters}{OT1}{cmss}{bx}{n}
\DeclareSymbolFont{ssfletters}{OT1}{cmss}{m}{n}
\DeclareMathSymbol{\bsfGamma}{0}{bsfletters}{'000}
\DeclareMathSymbol{\ssfGamma}{0}{ssfletters}{'000}
\DeclareMathSymbol{\bsfDelta}{0}{bsfletters}{'001}
\DeclareMathSymbol{\ssfDelta}{0}{ssfletters}{'001}
\DeclareMathSymbol{\bsfTheta}{0}{bsfletters}{'002}
\DeclareMathSymbol{\ssfTheta}{0}{ssfletters}{'002}
\DeclareMathSymbol{\bsfLambda}{0}{bsfletters}{'003}
\DeclareMathSymbol{\ssfLambda}{0}{ssfletters}{'003}
\DeclareMathSymbol{\bsfXi}{0}{bsfletters}{'004}
\DeclareMathSymbol{\ssfXi}{0}{ssfletters}{'004}
\DeclareMathSymbol{\bsfPi}{0}{bsfletters}{'005}
\DeclareMathSymbol{\ssfPi}{0}{ssfletters}{'005}
\DeclareMathSymbol{\bsfSigma}{0}{bsfletters}{'006}
\DeclareMathSymbol{\ssfSigma}{0}{ssfletters}{'006}
\DeclareMathSymbol{\bsfUpsilon}{0}{bsfletters}{'007}
\DeclareMathSymbol{\ssfUpsilon}{0}{ssfletters}{'007}
\DeclareMathSymbol{\bsfPhi}{0}{bsfletters}{'010}
\DeclareMathSymbol{\ssfPhi}{0}{ssfletters}{'010}
\DeclareMathSymbol{\bsfPsi}{0}{bsfletters}{'011}
\DeclareMathSymbol{\ssfPsi}{0}{ssfletters}{'011}
\DeclareMathSymbol{\bsfOmega}{0}{bsfletters}{'012}
\DeclareMathSymbol{\ssfOmega}{0}{ssfletters}{'012}

\newcommand{\bphi}{\bm{\phi}}

\newcommand{\bXi}{\bm{\Xi}}

\makeatletter
\newcommand*\rel@kern[1]{\kern#1\dimexpr\macc@kerna}
\newcommand*\widebar[1]{%
  \begingroup
  \def\mathaccent##1##2{%
    \rel@kern{0.8}%
    \overline{\rel@kern{-0.8}\macc@nucleus\rel@kern{0.2}}%
    \rel@kern{-0.2}%
  }%
  \macc@depth\@ne
  \let\math@bgroup\@empty \let\math@egroup\macc@set@skewchar
  \mathsurround\z@ \frozen@everymath{\mathgroup\macc@group\relax}%
  \macc@set@skewchar\relax
  \let\mathaccentV\macc@nested@a
  \macc@nested@a\relax111{#1}%
  \endgroup
}
\makeatother

\DeclareMathOperator*{\argmax}{arg\,max}

\DeclareMathOperator{\ST}{s.t.\ }

\DeclareMathOperator{\var}{var}

\DeclareMathOperator{\cov}{cov}

\newcommand{\ifbcdot}[1]{\ifblank{#1}{\cdot}{#1}}

\DeclarePairedDelimiterX\abs[1]{\lvert}{\rvert}{\ifbcdot{#1}}
\DeclarePairedDelimiterX\parens[1]{(}{)}{\ifbcdot{#1}}
\DeclarePairedDelimiterX\brk[1]{[}{]}{\ifbcdot{#1}}
\DeclarePairedDelimiterX\braces[1]{\{}{\}}{\ifbcdot{#1}}
\DeclarePairedDelimiterX\angles[1]{\langle}{\rangle}{\ifblank{#1}{\cdot,\cdot}{#1}}
\DeclarePairedDelimiterX\ip[2]{\langle}{\rangle}{\ifbcdot{#1},\ifbcdot{#2}}
\DeclarePairedDelimiterX\norm[1]{\lVert}{\rVert}{\ifbcdot{#1}}
\DeclarePairedDelimiterX\ceil[1]{\lceil}{\rceil}{\ifbcdot{#1}}
\DeclarePairedDelimiterX\floor[1]{\lfloor}{\rfloor}{\ifbcdot{#1}}

\DeclareFontFamily{U}{matha}{\hyphenchar\font45}
\DeclareFontShape{U}{matha}{m}{n}{
      <5> <6> <7> <8> <9> <10> gen * matha
      <10.95> matha10 <12> <14.4> <17.28> <20.74> <24.88> matha12
      }{}
\DeclareSymbolFont{matha}{U}{matha}{m}{n}
\DeclareFontSubstitution{U}{matha}{m}{n}

\DeclareFontFamily{U}{mathx}{\hyphenchar\font45}
\DeclareFontShape{U}{mathx}{m}{n}{
      <5> <6> <7> <8> <9> <10>
      <10.95> <12> <14.4> <17.28> <20.74> <24.88>
      mathx10
      }{}
\DeclareSymbolFont{mathx}{U}{mathx}{m}{n}
\DeclareFontSubstitution{U}{mathx}{m}{n}

\DeclareMathDelimiter{\vvvert}{0}{matha}{"7E}{mathx}{"17}
\DeclarePairedDelimiterX\vertiii[1]{\vvvert}{\vvvert}{\ifbcdot{#1}}

\DeclarePairedDelimiterXPP\trace[1]{\operatorname{Tr}}{(}{)}{}{\ifbcdot{#1}} 
\DeclarePairedDelimiterXPP\col[1]{\operatorname{col}}{\{}{\}}{}{\ifbcdot{#1}} 
\DeclarePairedDelimiterXPP\row[1]{\operatorname{row}}{\{}{\}}{}{\ifbcdot{#1}} 
\DeclarePairedDelimiterXPP\erf[1]{\operatorname{erf}}{(}{)}{}{\ifbcdot{#1}}
\DeclarePairedDelimiterXPP\erfc[1]{\operatorname{erfc}}{(}{)}{}{\ifbcdot{#1}}
\DeclarePairedDelimiterXPP\KLD[2]{D}{(}{)}{}{\ifbcdot{#1}\, \delimsize\|\, \ifbcdot{#2}} 
\DeclarePairedDelimiterXPP\op[2]{\operatorname{#1}}{(}{)}{}{#2} 

\newcommand{\convp}{\stackrel{\mathrm{p}}{\longrightarrow}}

\DeclarePairedDelimiterXPP\indicate[1]{{\bf 1}}{\{}{\}}{}{\ifbcdot{#1}}

\NewDocumentCommand\ofrac{s m}{%
	\IfBooleanTF#1%
	{\dfrac{1}{#2}}%
	{\frac{1}{#2}}%
}
\NewDocumentCommand\ddfrac{s m m}{%
	\IfBooleanTF#1%
	{\dfrac{\mathrm{d} {#2}}{\mathrm{d} {#3}}}%
	{\frac{\mathrm{d} {#2}}{\mathrm{d} {#3}}}%
}
\NewDocumentCommand\ppfrac{s m m}{%
	\IfBooleanTF#1%
	{\dfrac{\partial {#2}}{\partial {#3}}}%
	{\frac{\partial {#2}}{\partial {#3}}}%
}

\newcommand{\setgiven}{:}
\providecommand\given{}

\DeclarePairedDelimiterX\Set[2]\{\}{%
	\if#1:%
		\renewcommand\given{\SetSymbol{:}}%
	\else%
		\renewcommand\given{\SetSymbol[\delimsize]{#1}}%
	\fi%
#2
}

\NewDocumentCommand\set{s O{\setgiven} m}{%
	\IfBooleanTF#1%
	{\Set*{#2}{#3}}%
	{\Set{#2}{#3}}%
}

\NewDocumentCommand{\evalat}{ s O{\big} m e{_^} }{%
\IfBooleanTF{#1}%
{\left. #3 \right|}{#3#2|}%
\IfValueT{#4}{_{#4}}%
\IfValueT{#5}{^{#5}}%
}

\providecommand\given{}
\DeclarePairedDelimiterXPP\cprob[1]{}(){}{
\renewcommand\given{\nonscript\,\delimsize\vert\allowbreak\nonscript\,\mathopen{}}%
#1%
}
\DeclarePairedDelimiterXPP\cexp[1]{}[]{}{
\renewcommand\given{\nonscript\,\delimsize\vert\allowbreak\nonscript\,\mathopen{}}%
#1%
}

\DeclareDocumentCommand \P { s e{_^} d() g } {%
	\mathbb{P}%
	\IfBooleanTF{#1}%
		{
			\IfValueT{#2}{_{#2}}%
			\IfValueT{#3}{^{#3}}%
			\IfValueTF{#5}{\cprob{#4 \given #5}}{\IfValueT{#4}{\cprob{#4}}}%
		}%
		{
			\IfValueT{#2}{_{#2}}%
			\IfValueT{#3}{^{#3}}%
			\IfValueTF{#5}{\cprob*{#4 \given #5}}{\IfValueT{#4}{\cprob*{#4}}}%
		}%
}

\DeclareDocumentCommand \E { s e{_^} o g } {%
	\mathbb{E}%
	\IfBooleanTF{#1}%
		{
			\IfValueT{#2}{_{#2}}%
			\IfValueT{#3}{^{#3}}%
			\IfValueTF{#5}{\cexp{#4 \given #5}}{\IfValueT{#4}{\cexp{#4}}}%
		}%
		{
			\IfValueT{#2}{_{#2}}%
			\IfValueT{#3}{^{#3}}%
			\IfValueTF{#5}{\cexp*{#4 \given #5}}{\IfValueT{#4}{\cexp*{#4}}}%
		}%
}

\DeclareDocumentCommand \Var { s e{_^} d() g } {%
	\var%
	\IfBooleanTF{#1}%
		{
			\IfValueT{#2}{_{#2}}%
			\IfValueT{#3}{^{#3}}%
			\IfValueTF{#5}{\cprob{#4 \given #5}}{\IfValueT{#4}{\cprob{#4}}}%
		}%
		{
			\IfValueT{#2}{_{#2}}%
			\IfValueT{#3}{^{#3}}%
			\IfValueTF{#5}{\cprob*{#4 \given #5}}{\IfValueT{#4}{\cprob*{#4}}}%
		}%
}

\DeclareDocumentCommand \Cov { s e{_^} d() g } {%
	\cov%
	\IfBooleanTF{#1}%
		{
			\IfValueT{#2}{_{#2}}%
			\IfValueT{#3}{^{#3}}%
			\IfValueTF{#5}{\cprob{#4 \given #5}}{\IfValueT{#4}{\cprob{#4}}}%
		}%
		{
			\IfValueT{#2}{_{#2}}%
			\IfValueT{#3}{^{#3}}%
			\IfValueTF{#5}{\cprob*{#4 \given #5}}{\IfValueT{#4}{\cprob*{#4}}}%
		}%
}

\ExplSyntaxOn
\NewDocumentCommand \dist {m o o} {%
\mathrm{#1}\left(%
	\IfValueT{#3}{%
		\tl_if_blank:nTF{ #3 }{\cdot\, \middle|\, }{#3\, \middle|\, }%
	}
	\IfValueT{#2}{#2}%
\right)%
}
\ExplSyntaxOff

\newcommand{\Bern}[1]{\dist{Bern}[#1]}
\newcommand{\Unif}[1]{\dist{Unif}[#1]}

\NewDocumentCommand {\cbrace} {t+ D[]{black} D(){\widthof{#5}} m m } {%
	\begingroup%
		\color{#2}
		\IfBooleanTF{#1}{%
			\overbrace{#4}^%
		}{
			\underbrace{#4}_%
		}%
		{\parbox[c]{#3}{\centering\footnotesize{#5}}}%
	\endgroup%
}

\let\oldforall\forall
\renewcommand{\forall}{\oldforall \, }

\let\oldexist\exists
\renewcommand{\exists}{\oldexist \, }

\makeatletter

\newcommand{\rankcolor}[2]{%
	\expandafter\renewcommand\csname #1\endcsname[1]{%
		\ifblank{##1}{%
			{\color{#2} \textbf{#2}}%
		}{%
			\ifmmode
				\textcolor{#2}{\bm{##1}}%
			\else%
				{\color{#2} \textbf{##1}}%
			\fi	
		}%
	}
}

\rankcolor{first}{red}
\rankcolor{second}{blue}
\rankcolor{third}{cyan}
\makeatother

\graphicspath{{./Figures/}{./figures/}}
\DeclareDocumentCommand{\includeCroppedPdf}{ o O{./Figures/} m }{
	\IfFileExists{#2#3-crop.pdf}{}{%
		\immediate\write18{pdfcrop #2#3.pdf #2#3-crop.pdf}}%
	\includegraphics[#1]{#2#3-crop.pdf}
}

\makeatletter
\newcommand*{\addFileDependency}[1]{
  \typeout{(#1)}
  \@addtofilelist{#1}
  \IfFileExists{#1}{}{\typeout{No file #1.}}
}
\makeatother

\definecolor{gray90}{gray}{0.9}
\def\colorlist{red,blue,brown,cyan,darkgray,gray,lightgray,green,lime,magenta,olive,orange,pink,purple,teal,violet,white,yellow}

\makeatletter
\def\startcomment{[}
\ifx\nohighlights\undefined
	\newcommand{\createcolor}[1]{%
			\expandafter\newcommand\csname #1\endcsname[1]{{\color{#1} ##1}}%
	}
	\newcommand{\msout}[1]{\text{\color{green} \st{\ensuremath{#1}}}}
	\newcommand{\del}[1]{{\color{green}\ifmmode \msout{#1}\else\st{#1}\fi}}
\else
	\newcommand{\createcolor}[1]{%
			\expandafter\newcommand\csname #1\endcsname[1]{%
				\noexpandarg%
				\StrChar{##1}{1}[\firstletter]%
				\if\firstletter\startcomment%
					\relax
				\else%
					##1
				\fi
			}%
	}
	\newcommand{\msout}[1]{}
	\newcommand{\del}[1]{}
\fi

\def\@tempa#1,{%
    \ifx\relax#1\relax\else
        \createcolor{#1}%
        \expandafter\@tempa
    \fi
}
\expandafter\@tempa\colorlist,\relax,
\makeatother

\newcommand{\hhide}[1]{}

\ifx\diagnoselabel\undefined
	\relax
\else
	\makeatletter
	\def\@testdef #1#2#3{%
		\def\reserved@a{#3}\expandafter \ifx \csname #1@#2\endcsname
			\reserved@a  \else
			\typeout{^^Jlabel #2 changed:^^J%
				\meaning\reserved@a^^J%
				\expandafter\meaning\csname #1@#2\endcsname^^J}%
			\@tempswatrue \fi}
	\makeatother
\fi

\def\BibTeX{{\rm B\kern-.05em{\sc i\kern-.025em b}\kern-.08em
    T\kern-.1667em\lower.7ex\hbox{E}\kern-.125emX}}

\newacronym{GGSP}{GGSP}{generalized graph signal processing}
\newacronym{GSO}{GSO}{graph shift operator}
\newacronym{GSP}{GSP}{graph signal processing}
\newacronym{MHT}{MHT}{multiple hypothesis testing}
\newacronym{FDR}{FDR}{false discovery rate}
\newacronym{lfdr}{lfdr}{local false discovery rate}
\newacronym{pdf}{pdf}{probability density function}
\newacronym{cdf}{cdf}{cumulative distribution function}
\newacronym{KKT}{KKT}{Karush–Kuhn–Tucker}
\newacronym{WLLN}{WLLN}{weak law of large numbers}
\newacronym{CSLLN}{CSLLN}{conditional strong law of large numbers}
\newacronym{DCT}{DCT}{dominated convergence theorem}
\newacronym{BIC}{BIC}{Bayesian information criterion}
\newacronym{MLE}{MLE}{maximum likelihood estimator}
\newacronym{a.e.}{a.e.}{almost everywhere}
\newacronym{a.s.}{a.s.}{almost surely}
\newacronym{r.c.d.}{r.c.d.}{regular conditional distribution}
\newacronym{AWGN}{AWGN}{additive white Gaussian noise}
\newacronym{BH}{BH}{Benjamini-Hochberg}
\newacronym{GLM}{GLM}{generalized linear model}
\newacronym{ADMM}{ADMM}{alternating direction method of multipliers}

\newcommand{\FDR}{\mathrm{FDR}}
\newcommand{\mFDR}{\mathrm{mFDR}}
\newcommand{\pow}{\mathrm{pow}}
\newcommand{\mpow}{\mathrm{mpow}}
\newcommand{\lfdr}{\mathrm{lfdr}}
\newcommand{\Ind}{\mathbb I}

\newcommand{\sigmoid}{\mathrm{sigmoid}}

\begin{document}

\title{A Generalized Graph Signal Processing Framework for Multiple Hypothesis Testing over Networks\\
\thanks{This research/project is supported by the National Research Foundation, Singapore and Infocomm Media Development Authority under its Future Communications Research \& Development Programme.}
}


\author{
    \IEEEauthorblockN{Xingchao Jian\IEEEauthorrefmark{1}, Martin Gölz\IEEEauthorrefmark{2}, Feng Ji\IEEEauthorrefmark{1}, Wee~Peng~Tay\IEEEauthorrefmark{1}, Abdelhak M. Zoubir\IEEEauthorrefmark{2}}
    \IEEEauthorblockA{\IEEEauthorrefmark{1}School of Electrical and Electronic Engineering, Nanyang Technological University, Singapore
    }
    \IEEEauthorblockA{\IEEEauthorrefmark{2}Technische Universitat Darmstadt, Darmstadt, Germany
    }
}

\maketitle

\begin{abstract}
We consider the multiple hypothesis testing (MHT) problem over the joint domain formed by a graph and a measure space. On each sample point of this joint domain, we assign a hypothesis test and a corresponding $p$-value. The goal is to make decisions for all hypotheses simultaneously, using all available $p$-values. In practice, this problem resembles the detection problem over a sensor network during a period of time. To solve this problem, we extend the traditional two-groups model such that the prior probability of the null hypothesis and the alternative distribution of $p$-values can be inhomogeneous over the joint domain. We model the inhomogeneity via a generalized graph signal. This more flexible statistical model yields a more powerful detection strategy by leveraging the information from the joint domain.
\end{abstract}


\section{Introduction}

Detection is an important task in statistical signal processing. \Gls{MHT}, which aims to make decisions simultaneously for many hypotheses, has been considered since the nominal works \cite{BenHoc:J95,Sto:J02}. The detection strategy is expected to discover as many alternatives as possible, while maintaining the \gls{FDR} below a nominal level. When the hypotheses are placed on a known structure, it is expected that the detection power could be boosted by leveraging the underlying structure. There are works that address this problem when the underlying structure is a graph \cite{LiBar:J18,TanKoyPolSco:J18,PouXia:J23} or spatial domain \cite{GolZouKoi:J22}. In this paper, we consider the \gls{MHT} problem on a joint graph-measure space structure by utilizing the basic concepts in \gls{GGSP}.

\Gls{GSP} is a set of techniques that address signal processing problems on graphs, which mainly includes sampling \cite{AniGadOrt:J16, TanEldOrt:J20}, reconstruction \cite{RomMaGia:J16,KroRouEld:J22} and filtering \cite{PavGirOrt:J23, IsuGamShuSeg:J24}. In \gls{GSP}, signals are functions on the vertex set of a graph. This signal model is extended by \cite{JiTay:C18,JiTay:J19} to a more general setting where the domain of signal is the Cartesian product between the vertex set and a general measure space. We refer to this product domain as \emph{joint domain}, which resembles the joint time-vertex domain. The signal reconstruction and filtering problem for \gls{GGSP} has been addressed in \cite{JianTayEld:J23,JianTay:C23}. In this work, we consider the \gls{MHT} problem on the joint domain. We propose to parameterize the distributions of $p$-values on graphs by a generalized graph signal. This strategy is different from the existing works in that we model and estimate continuously varying \glspl{pdf} of $p$-values over the domain, instead of homogeneous or having a fixed number of choices. Therefore, this strategy is better at leveraging the domain knowledge than the existing approaches.

The contributions of this paper are as follows.

\begin{enumerate}
	\item We formulate the empirical-Bayesian model for \gls{MHT} over the joint domain, and derive the optimal oracle detection strategy.
	\item We provide a method for estimating probability densities of $p$-values varying over the joint domain. By substituting this estimate into the oracle solution, we propose a detection strategy for \gls{MHT} problem with \gls{FDR} control over the joint domain.
	\item We illustrate our method on the seismic dataset.
\end{enumerate}
The rest of this paper is organized as follows. In \cref{sect:prob_for}, we introduce the problem formulation of \gls{MHT} on the joint domain. In \cref{sect:stat_model_sol}, we describe the empirical-Bayesian model, from which we derive the oracle solution given the true distribution of $p$-values. In \cref{sect:est_lfdr}, we prove the consistency of \gls{MLE} estimator for the underlying generalized graph signal. By substituting this estimator into the oracle solution, we obtain a strategy with asymptotic \gls{FDR} control.

\emph{Notations.} We write random variables in serif font for English letters (e.g., $\sfJ$, $\sfS$ $\sfv$ and $\sft$, etc.) or teletyped font (e.g., $\tttheta$, $\tteta$, $\ttxi$, $\ttgamma$) for Greek letters. We use plain lowercase letters (e.g., $\phi$) to represent scalars. Bold lowercase letters (e.g., $\bphi$) stand for vectors and vector-valued functions. Bold upper-case letters (e.g., $\bS$) stand for matrices. We write $\convp$ for convergence in probability. For a set $\calX$, we use $\calX^M$ to denote $M$-times Cartesian product of the set $\calX$ with itself.

\section{Problem Formulation}\label{sect:prob_for}

In this section, we introduce the mathematical formulation of the \gls{MHT} problem. We first introduce the construction of the domain where the problem is based on. We consider an undirected graph $G = (\calV,\calE)$ where $\calV$ denotes the vertex set and $\calE$ denotes the edge set. We assume that $G$ does not have self-loops. Apart from the graph, we consider an additional domain $(\calT,\calA,\tau)$ which is a measure space. The domain of the problem is then $\calV\times\calT$. For simplicity, we write $\calJ$ to denote $\calV\times\calT$. Suppose $\abs{\calV}=N$. We denote the \gls{GSO} as $\bS_G$ and its corresponding graph Fourier basis as $\set{\phi_k(\cdot)\given k=1,\dots,N}$. Note that each $\phi_k(\cdot)$ is a function on $\calV$.

The \gls{MHT} problem aims to make decisions based on $p$-values observed by multiple hypotheses simultaneously. In this paper, we consider a finite sample set of $\calJ$, denoted as $\sfS := \set{(\sfv_m,\sft_m)\given m=1,\dots,M}$. Here $\sfS$ can be a random set. Each sample $(\sfv_m,\sft_m)\in\sfS$ is associated with a hypothesis and a $p$-value $\sfp_m$. The set of sample points with null hypotheses is $\sfJ_0:=\set{(\sfv_m,\sft_m)\given \text{null hypothesis is true}}$, and $\sfJ_1:=\sfS\setminus\sfJ_0$ denotes those with alternative hypotheses. Let $\bp:=(\sfp_1,\dots,\sfp_M)$ be the vector containing all $p$-values. We consider a \emph{detection strategy} as a map $\bh: [0,1]^M\to \set{0,1}^M$, so that the value of $\bh$ is a $M$-dimensional vector with $0$-$1$ entries. If the $m$-th entry of $\bh(\bp)$ is $0$ (or $1$), then it means that the hypothesis on $(\sfv_m,\sft_m)$ is identified as null (or alternative) by $\bh$. We write $\widehat{\sfJ}_0$ (or $\widehat{\sfJ}_1$) as the sample points where the associated hypotheses are identified as null (or alternative) by $\bh$. There are two measurements for the quality of a detection strategy: \gls{FDR} and power. They are defined as follows:
\begin{align}
	\FDR(\bh) &= \E[\frac{\abs{\widehat{\sfJ}_1\bigcap \sfJ_0}}{\max(\abs{\widehat{\sfJ}_1}, 1)}],\label{eq:def_fdr}\\
	\pow(\bh) &= \E[\frac{\abs{\widehat{\sfJ}_1\bigcap \sfJ_1}}{\max(\abs{\sfJ_1},1)}].\label{eq:def_pow}
\end{align}
Similar to a single hypothesis test, the target of \gls{MHT} is to find a way to design $\bh$ whose power is as large as possible while the \gls{FDR} is controlled below a nominal level.

\section{Statistical Model and Oracle Solution}\label{sect:stat_model_sol}

\subsection{Statistical Model}

In this section, we introduce the empirical-Bayesian model for the \gls{MHT} problem described in \cref{sect:prob_for}. This model extends what is known as the two-groups model in traditional \gls{MHT}, which has two main assumptions: first, each hypothesis being null or alternative is Bernoulli with the same distribution. Second, under null and alternative hypothesis, the $p$-value has different distributions. In this paper, since the hypotheses are placed on different sample points on $\calJ$, we allow the prior probability of being null and the distributions of $p$-values to vary among different points of $\calJ$. By doing so, the model has better fitness on the data than the existing models in \cite{PouXia:J23,GolZouKoi:J22,LiBar:J18,TanEldOrt:J20}.

We illustrate the hierarchical Bayes model in \cref{fig:Bayes:model}. 
\begin{figure}[!htb]
	\centering
	\includegraphics[scale = 1]{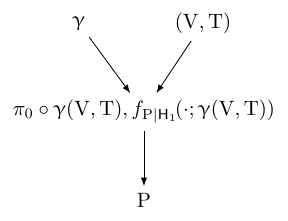}
	\caption{The scheme of the Bayesian model for \gls{MHT} in \gls{GGSP}.}
	\label{fig:Bayes:model}
\end{figure}
In this model, $\ttgamma$ is a stochastic process on $\calJ$. We use it to parameterize the varying alternative distribution of $p$-values and the probability of null hypothesis for different $(v,t)\in\calJ$. We assume that the range of $\ttgamma$ as a function of $(v,t)$ is a subset of $\calZ$.
A sample point $(\sfv, \sft)$ is drawn from a positive probability measure $\rho$ on $\calJ$, independently with 
the process $\ttgamma$. 
Let $\pi_0$ be a continuous function from $\calZ$ to $[0,1]$. Given $\ttgamma$ and $(\sfv,\sft)$, we obtain a probability $\pi_0\circ\ttgamma(\sfv,\sft)\in[0,1]$. 
We assume the following hierarchical model for generating the $p$-value: first, generate $\tttheta$ as
\begin{align*}
	\tttheta=
	\begin{cases}
		0 & \mathrm{w.p.~} \pi_0\circ\ttgamma(\sfv,\sft), \\
		1 & \mathrm{w.p.~} 1-\pi_0\circ\ttgamma(\sfv,\sft),
	\end{cases}
\end{align*}
where $\tttheta=1$ indicates that the alternative hypothesis on $(\sfv,\sft)$ is true. Then, under null and alternative hypotheses, the $p$-value $\sfp$ is generated independently from the following different distributions:
\begin{align*}
	\sfp\mid \set{\tttheta=0} &\sim f_0(\cdot;(\sfv,\sft)),\\
	\sfp\mid \set{\tttheta=1} &\sim f_1(\cdot;\ttgamma(\sfv,\sft)),
\end{align*}
where $f_1$ is from a parametric family of \glspl{pdf} $\calP^1:=\set{f_1(\cdot;\zeta)\given \zeta\in\calZ}$ on $[0,1]$ with parameter $\zeta$. Under this model, given $\ttgamma$, the conditional joint \gls{pdf} of the $p$-value $\sfp$, the indicator $\tttheta$, and the sample point $(\sfv,\sft)$ is 
\begin{align*}
	f_{\sfp, \tttheta, (\sfv,\sft)}(p,\theta,(v,t)\mid \ttgamma) 
	= \Big(\pi_0\circ\ttgamma(v,t)\Ind\set{\theta=0}f_0(p;(v,t)) \\
	+ (1-\pi_0\circ\ttgamma(v,t))\Ind\set{\theta=1}f_1(p;\ttgamma(v,t)) \Big)\rho(v,t),
\end{align*}
where $p\in(0,1]$, $\theta\in\set{0,1}$, and $(v,t)\in\calJ$. From $f_{\sfp, \tttheta, (\sfv,\sft)}$, we derive the conditional \gls{pdf} of $\sfp$ and $(\sfv,\sft)$ as
\begin{align}
	\begin{aligned}\label{eq:pdf_pvt|gam}
	f_{\sfp,(\sfv,\sft)}(p,(v,t)\mid\ttgamma) 
	= \Big(\pi_0\circ\ttgamma(v,t)f_0(p;(v,t)) \\
	+ (1-\pi_0\circ\ttgamma(v,t))f_1(p;\ttgamma(v,t)) \Big)\rho(v,t),
	\end{aligned}
\end{align}
and the conditional \gls{pdf} of $\sfp$ given $\ttgamma$ and $(\sfv,\sft)$ as
\begin{align}
	\begin{aligned}\label{eq:pdf_p|gamvt}
	f_{\mathrm{mix}}(p\mid\ttgamma(v,t))
	&= \pi_0\circ\ttgamma(v,t)f_0(p;(v,t)) \\
	&+ (1-\pi_0\circ\ttgamma(v,t))f_1(p;\ttgamma(v,t)).
	\end{aligned}
\end{align}

In other words, we parameterize the distribution of $p$-values on $\calJ$ by a \emph{random generalized graph signal} $\ttgamma$. This model combines the information from the joint domain $\calJ$ via $\ttgamma$ and the two-groups model widely adopted in the \gls{MHT} literature \cite{Efr:10,LeiFit:J18,CaoChenZhang:J22}.

In this paper, we assume that $f_0(\cdot;(v,t))$ is known and deterministic for all $(v,t)\in\calJ$. By definition of $p$-value, it can be verified that $f_0(\cdot;(v,t))$ under simple null hypothesis is $\Unif{0,1}$. In other cases when $f_0(\cdot;(v,t))$ is unknown, $f_0(\cdot;(v,t))$ may be estimated from data beforehand by controlling the null hypotheses to be true. 
For the $M$ \gls{iid} samples $(\sfv_m,\sft_m)$ in $\sfS$, we obtain $M$ independent samples $(\tttheta_m)_{m=1}^M$, and $p$-values $(\sfp_m)_{m=1}^M$.

We note that this model generalizes the two-groups model in the \gls{MHT} literature. In traditional \gls{MHT},  each $\tttheta_m$ is modeled as a Bernoulli random variable $\Bern{1 - \pi_0}$ where $\pi_0$ is a fixed value. The distributions $f_0$ and $f_1$ are also assumed to be the same for all $p$-values. In this model, by parameterizing $f_1$ and $\pi_0$ with $\ttgamma(v,t)$, we indicate that both the probability of hypotheses being null and the distribution of alternative $p$-values are inhomogeneous over $\calJ$. 

In this paper, we assume that $\pi_0\circ\ttgamma(v,t)$ and $f_1(\cdot;\ttgamma(v,t))$ are identifiable from $f_{\mathrm{mix}}(\cdot\mid\ttgamma(v,t))$ using \cref{eq:pdf_p|gamvt}. This can be guaranteed by the following assumptions:
\begin{Assumption}\label{asp:identify}\
	\begin{enumerate}[(i)]
		\item \label[condition]{cond:monot_f_0} $f_0(p;(v,t))$ is non-decreasing \gls{wrt} $p\in[0,1]$.
		\item \label[condition]{cond:monot_f_1}$f_1(\cdot;\zeta)$ is non-increasing on $(0,1]$ for all $\zeta\in\calZ$.
		\item \label[condition]{cond:f_1=0}$\min\limits_{p\in[0,1]}f_1(p;\zeta)=0$ for all $\zeta\in\calZ$, i.e., with \cref{cond:monot_f_1}, $f_1(1;\zeta)=0$ for all $\zeta\in\calZ$.
		\item \label[condition]{cond:conts_pdf}$f_0(p;(v,t))$ is continuous on $[0,1]\times\calJ$, and $f_1(p;\zeta)$ is continuous on $(0,1]\times\calZ$.
	\end{enumerate}
\end{Assumption}

In \cref{asp:identify}, \cref{cond:monot_f_0,cond:monot_f_1} indicate that $p$-value is more likely to be small under the alternative hypothesis and more likely to be large under the null hypothesis. 
The assumptions in \cref{asp:identify} are commonly assumed in the \gls{MHT} literature (cf .\ \cite[Theorem 2]{LeiFit:J18} and \cite[Section 2.1]{CaoChenZhang:J22}).
They ensure that $\pi_0\circ\ttgamma(v,t)$ and $f_1(\cdot;\ttgamma(v,t))$ are identifiable from $f_{\mathrm{mix}}(\cdot\mid\ttgamma(v,t))$, since according to \cref{cond:monot_f_0,cond:monot_f_1,cond:f_1=0}, for each $(v,t)\in\calJ$, we have  
\begin{align}
	\pi_0\circ\ttgamma(v,t) &= \frac{f_{\mathrm{mix}}(1\mid\ttgamma(v,t))}{f_0(1;(v,t))},\label{eq:pi0_inv}\\
	f_1(p;\ttgamma(v,t)) &= \ofrac{1-\pi_0\circ\ttgamma(v,t)} \Big(f_{\mathrm{mix}}(p\mid\ttgamma(v,t)) \nn
	&- \pi_0\circ\ttgamma(v,t)f_0(p;(v,t)) \Big).\label{eq:f1_inv}
\end{align}

The \gls{lfdr} is defined as 
\begin{align}\label{eq:lfdr}
	\lfdr(p; \ttgamma(v,t)) = \frac{\pi_0\circ\ttgamma(v,t)f_0(p;(v,t))}{f_{\mathrm{mix}}(p\mid\ttgamma(v,t))}.
\end{align}

The \gls{lfdr} for $\sfp_m,(\sfv_m,\sft_m)$ is the conditional probability $\P(\tttheta_m=0\given \ttgamma(\sfv_m,\sft_m),\sfp_m)$.
In \cref{sect:est_lfdr}, we estimate $f_{\mathrm{mix}}(\cdot\mid\ttgamma(\sfv_m,\sft_m))$, and thus $\pi_0\circ\ttgamma(\sfv_m,\sft_m)$ and $f_1(\cdot;\ttgamma(\sfv_m,\sft_m))$. 

\subsection{Oracle Solution}

In this subsection, we show that, if $\pi_0\circ\ttgamma(\sfv_m,\sft_m)$ and $f_1(\cdot;\ttgamma(\sfv_m,\sft_m))$ are known for $m=1,\dots,M$, then the optimal detection strategy is thresholding the \gls{lfdr}. To see this, we introduce two complementary definitions, marginal \gls{FDR} and marginal power, defined as 
\begin{align}
	\mFDR(\bh;\ttgamma,\sfS) &= \frac{\E[\abs{\widehat{\sfJ}_1\bigcap \sfJ_0} \given \ttgamma(\sfS)]}{\E[\abs{\widehat{\sfJ}_1}\given\ttgamma(\sfS)]},\label{eq:def_mfdr}\\
	\mpow(\bh;\ttgamma,\sfS) &= \frac{\E[\abs{\widehat{\sfJ}_1\bigcap \sfJ_1} \given \ttgamma(\sfS)]}{\E[\abs{\sfJ_1}\given\ttgamma(\sfS)]}.\label{eq:def_mpow}
\end{align}

In classical single hypothesis test, the null hypothesis is rejected if the $p$-value is lower than the nominal significance level. In this paper, for each individual test on $(\sfv_m,\sft_m)$, we use the following thresholding rule:
\begin{align}\label{eq:thres_rule}
	(\bh(\bsfp))_m = 
	\begin{cases}
		1 & \text{if}\ \sfp_m \leq s_m, \\
		0 & \text{otherwise}.
	\end{cases}
\end{align}
The difference is that, in \cref{eq:thres_rule}, the threshold of $p$-value is no longer the significance level but chosen by certain criteria that we will soon introduce. Under this rule, designing the detection strategy $\bh$ amounts to designing $\bs:=(s_1,\dots,s_M)$. Since the rejection rule $\bh$ is fully determined by $\bs$, we denote $\mFDR(\bh; \ttgamma,\sfS)$ and $\mpow(\bh; \ttgamma,\sfS)$ by $\mFDR(\bs; \ttgamma,\sfS)$ and $\mpow(\bs; \ttgamma,\sfS)$ under \cref{eq:thres_rule}. We thus consider the following problem:
\begin{align}\label{eq:opt_mpow}
	\begin{aligned}
		\max_{\bs\in[0,1]^M} \mpow(\bs; \ttgamma,\sfS) \\
		\ST \mFDR(\bs; \ttgamma,\sfS) \leq \alpha.
	\end{aligned}
\end{align}
In \cref{thm:opt_thres}, we show that the optimal solution to problem \cref{eq:opt_mpow} is a level surface of \gls{lfdr}. This theorem is a slight generalization of \cite[Theorem 2]{LeiFit:J18}.
\begin{Theorem}\label{thm:opt_thres}
	Suppose \cref{asp:identify} holds and there exists $p\in(0,1)$ and $(v,t)\in\sfS$ such that $\lfdr(p;\ttgamma(v,t))< \alpha$.
	Then the optimal solution $\bs^*=(s_1^*,\dots,s_M^*)$ to problem \cref{eq:opt_mpow} satisfies
	\begin{align*}
		\lfdr(s_m^*;\ttgamma(\sfv_m,\sft_m)) = \eta,\ m=1,\dots,M,
	\end{align*}
	where $\eta$ is a constant independent of $m$.
\end{Theorem}
\begin{proof}
	The proof of this theorem is a slight modification of the proof of \cite[Theorem 3]{LeiFit:J18}. We refer the reader to \cite[Appendix B]{JianGolJi:J24} for the proof details.
\end{proof}

Different from \cite[Theorem 2]{LeiFit:J18} and \cite[Proposition 2.1]{CaoChenZhang:J22}, in \cref{thm:opt_thres}, we do not require that $f_1(\cdot;\ttgamma(v,t))$ is continuous on the \emph{closed interval} $[0,1]$. This means that we allow for unbounded $f_1$ such as the Beta distribution. \cref{thm:opt_thres} implies that the optimal threshold for $\set{\sfp_m}$ corresponds to a level set of \gls{lfdr}. Therefore, the rejection rule becomes  
\begin{align}\label{eq:orac_thres}
	(\bh(\bsfp))_m=
	\begin{cases}
		1 & \lfdr(\sfp_m; \ttgamma(\sfv_m,\sft_m)) \leq \eta, \\
		0 & \lfdr(\sfp_m; \ttgamma(\sfv_m,\sft_m)) > \eta.
	\end{cases}
\end{align}
To choose $\eta$, we define the following quantities:
\begin{align}
	\sfd_{1,M}(\eta) &:= \ofrac{M}\sum_{m=1}^M\lfdr(\sfp_m;\ttgamma(\sfv_m,\sft_m))\cdot \nn
	&\Ind\set{\lfdr(\sfp_m;\ttgamma(\sfv_m,\sft_m))\leq \eta},\label{eq:def:D_1M}\\
	\sfd'_{1,M}(\eta) &:=\ofrac{M}\sum_{m=1}^M (1-\tttheta_m)\Ind\set{\lfdr(\sfp_m;\ttgamma(\sfv_m,\sft_m))\leq \eta}\label{eq:def:V_M}\\
	\sfd_{0,M}(\eta) &:= \ofrac{M}\sum_{m=1}^M\Ind\set{\lfdr(\sfp_m;\ttgamma(\sfv_m,\sft_m))\leq \eta}.\label{eq:def:D_0M}
\end{align}

By conditioning on $\sfp_m$ and noting that $\P(\tttheta_m=0\given \sfp_m,\ttgamma(\sfv_m,\sft_m)) = \lfdr(\sfp_m;\ttgamma(\sfv_m,\sft_m))$, we see that $\sfd_{1,M}(\eta)$ estimates the proportion of false rejections among all tests:
\begin{align*}
	&\E[(1-\tttheta_m)\Ind\set{\lfdr(\sfp_m;\ttgamma(\sfv_m,\sft_m))\leq \eta}\given\ttgamma,\sfS]\\
	&= \E[\lfdr(\sfp_m;\ttgamma(\sfv_m,\sft_m))\Ind\set{\lfdr(\sfp_m;\ttgamma(\sfv_m,\sft_m))\leq \eta}\given\ttgamma,\sfS].
\end{align*}
Therefore, by taking expectation over $\sfS$, we see that $\E[\sfd_{1,M}(\eta)\given\ttgamma]=\E[\abs{\widehat{\sfJ}_1\bigcap\sfJ_0}\given\ttgamma]=\E[\sfd'_{1,M}(\eta)\given\ttgamma]$. Besides, it can be shown that $\E[\sfd_{0,M}(\eta)\given\ttgamma]=\E[\abs{\widehat{\sfJ}_1}\given\ttgamma]$. Therefore, the quantity $\sfr_M(\eta):=\dfrac{\sfd_{1,M}(\eta)}{\sfd_{0,M}(\eta)}$
approximates $\mFDR(\bh;\ttgamma,\sfS)$ in \cref{eq:def_mfdr}. On the other hand, note that $\mpow(\bh;\ttgamma,\sfS)$ in \cref{eq:def_mpow} increases with $\eta$. Therefore, we choose the optimal $\eta$ (denoted by $\tteta_M$) in the following way:
\begin{align}\label{eq:orac_sol}
	\tteta_M:=\sup\set{\eta\given \sfr_M(\eta)\leq\alpha}.
\end{align}
Note that in practice we usually do not have access to the ground truths $\pi_0\circ\ttgamma(\sfv_m,\sft_m)$ and $f_1(\cdot;\ttgamma(\sfv_m,\sft_m))$, hence the \gls{lfdr} is unknown. Therefore, we call \cref{eq:orac_thres,eq:orac_sol} the \emph{oracle} solution.

\section{Estimation of Local False Discovery Rate}\label{sect:est_lfdr}

In this section, we state our method of substituting the \gls{MLE} estimator to the oracle solution. Specifically, we first estimate $\ttgamma$, and thus obtain an estimate of $\sfr_M$. To estimate $\ttgamma$, we consider it as a band-limited generalized graph signal.
We make the following assumptions to ensure the consistency of the \gls{MLE} estimation of $\ttgamma$.
\begin{Assumption}\label{asp:bandlimit}\
	\begin{enumerate}[(i)]
		\item The signal $\ttgamma$ is bandlimited, i.e., for all $(v,t)\in\calJ$,
		\begin{align}\label{eq:bl_para}
			\ttgamma(v,t) = \sum_{k_1=1}^{K_1}\sum_{k_2=1}^{K_2} \ttxi_{k_1,k_2} \cdot\phi_{k_1}(v)\psi_{k_2}(t),
		\end{align}
		where $K_1, K_2$ are positive integers, $\set{\phi_k(\cdot)\given k=1,\dots,N}$ is the graph Fourier basis and $\set{\psi_k(\cdot)\given k\in\bbN}$ is a set of orthonormal basis of $L^2(\calT)$. The coefficient matrix $\bttXi:=(\ttxi_{k_1,k_2})\in\Real^{K_1\times K_2}$ is a random matrix and takes values in a convex and compact set $\calK\subset\Real^{K_1\times K_2}$. Under this assumption, we write $\ttgamma(v,t)$ as $\gamma(v,t;\bttXi)$ to highlight the relationship \cref{eq:bl_para}.
		
		\item \label[condition]{cond:cont_basis} The function $\phi_{k_1}(v)\psi_{k_2}(t)$ is continuous in $(v,t)\in\calJ$ for all $1\leq k_1\leq K_1,1\leq k_2\leq K_2$.
		
		\item \label[condition]{cond:identi_pvt} For any distinct $\bXi\neq\bXi'$, $f_{\mathrm{mix}}(p\mid\gamma(v,t;\bXi))\neq f_{\mathrm{mix}}(p\mid\gamma(v,t;\bXi'))$ on a set in $(0,1]\times\calJ$ with positive measure. We denote the distribution and expectation of $\sfp,(\sfv,\sft)$ under $\bXi$ by $\P_{\bXi}$ and $\E_\bXi$.
		
		\item For all $\bXi$,
		\begin{align}
			\E_{\bXi}[\ln\norm*{\dfrac{f_{\mathrm{mix}}(\sfp\mid\gamma(\sfv,\sft;\bXi'))}{f_{\mathrm{mix}}(\sfp\mid\gamma(\sfv,\sft;\bXi))}}_\infty]<\infty,
		\end{align}
		where the sup norm is taken \gls{wrt} $\bXi'$, and the expectation is taken \gls{wrt} $\sfp_m,(\sfv_m,\sft_m)$.
	\end{enumerate}
\end{Assumption}
Under \cref{asp:bandlimit}, we know that $\ttgamma$ is parameterized by the Fourier coefficients $\bttXi$. We use the \gls{MLE} method to estimate them.
Note that $\rho(\sfv_m,\sft_m)$ does not depend on $\bXi$, so the \gls{MLE} $\widehat{\bttXi}$ can be obtained by 
\begin{align}\label{eq:MLE_prob}
	\widehat{\bttXi} &= \argmax_{\bXi\in\calK} \sum_{m=1}^M \ln f_{\mathrm{mix}}(\sfp_m\mid\gamma(\sfv_m,\sft_m;\bXi))+ \ln\rho(\sfv_m,\sft_m) \nn
	&= \argmax_{\bXi\in\calK} \sum_{m=1}^M \ln f_{\mathrm{mix}}(\sfp_m\mid\gamma(\sfv_m,\sft_m;\bXi)).
\end{align}
By the consistency of \glspl{MLE}, $\widehat{\bttXi}$ converges to $\bttXi$ in probability as $M\to\infty$, hence $\widehat{\ttgamma}(v,t) = \sum_{k_1=1}^{K_1}\sum_{k_2=1}^{K_2} \widehat{\ttxi}_{k_1,k_2} \cdot\phi_{k_1}(v)\psi_{k_2}(t)$
converges to $\ttgamma(v,t)$ in probability. We formally state this in the following theorem:

\begin{Theorem}\label{thm:unif_conv_para}
	Under \cref{asp:bandlimit}, we have $\widehat{\bttXi}\convp\bttXi$ as the number of samples $M\to\infty$ and
	\begin{align}\label{eq:det_consis_mle}
		\sup_{(v,t)\in\calJ}\abs{\widehat{\ttgamma}(v,t) - \ttgamma(v,t)}\convp0,
	\end{align}
	under the probability measure conditioned on $\bttXi$.
\end{Theorem}
\begin{proof}
	The proof of this theorem is in \cite[Appendix C]{JianGolJi:J24}.
\end{proof}
Once we have estimated $\ttgamma(v,t)$,  we can then estimate \gls{lfdr} by replacing $\ttgamma(v,t)$ with $\gamma(v,t;\widehat{\bttXi})$ in \cref{eq:lfdr}, i.e., $\lfdr(p;\gamma(v,t;\widehat{\bttXi}))$. By replacing the true \gls{lfdr} with this estimate in $\sfd_{1,M}(\eta)$ and $\sfd_{0,M}(\eta)$, we obtain $\widehat{\sfd}_{1,M}(\eta)$ and $\widehat{\sfd}_{0,M}(\eta)$. Finally we define $\widehat{\sfr}_M(\eta):=\frac{\widehat{\sfd}_{1,M}(\eta)}{\widehat{\sfd}_{0,M}(\eta)}$. Therefore, by \cref{eq:orac_sol}, we design the rejection threshold $s_m$ such that
\begin{align*}
	\lfdr(s_m;\gamma(\sfv_m,\sft_m;\widehat{\bttXi})) = \widehat{\tteta}_M,
\end{align*}
where $\widehat{\tteta}_M:=\sup\set{\eta\given\widehat{\sfr}_M(\eta)\leq\alpha}$.
In practice, we need to choose the hyperparameters $K_1$ and $K_2$ to balance the goodness of fit and model complexity. Let $l_{K_1,K_2}^*$ be the optimal value of \cref{eq:MLE_prob}. We propose to choose the optimal pair of $(K_1,K_2)$ with the smallest \gls{BIC} for choosing these parameters:
\begin{align*}
	\mathrm{BIC} = K_1K_2\ln M - 2l_{K_1,K_2}^*.
\end{align*}
We name this detection strategy as MHT-GGSP. It can be proved that, under mild conditions, this strategy admits asymptotic \gls{FDR} control \cite[Theorem 3]{JianGolJi:J24}. 

\section{Numerical Results}

In this section, we consider the seismic signal detection task in a sensor network. The seismic event occurred at 17:32:40, 22 December 2020 UTC in New Zealand, with the origin latitude and longitude (-42.14, 171.94).\footnote{https://www.geonet.org.nz/data/access/FDSN} We focus on the $32$ stations that are either within a radius of $4$ degrees from the origin or recorded this event. The sensor network is constructed as a $3$-NN graph based on the stations' latitudes and longitudes. We are interested in a time window of 60 seconds before and 60 seconds after the event. For each second, each station tests the presence of a seismic signal using the $z$-detector \cite{Ber:80}. To compute the signal energy, we first calculate the average of the squared values of the observed waveform in a one-second time window, and then take the logarithm. We denote this quantity as $\log\mathrm{STA}$. To compute the mean and standard deviation of $\log\mathrm{STA}$, we use the period from 7200 seconds to 60 seconds before the event time as the long-term history. Finally, the $z$-detector for each second is computed by normalizing $\log\mathrm{STA}$ using the mean and standard deviation. Under the null hypothesis (no seismic signal present), the $z$-detector is assumed to follow a standard normal distribution. Therefore, a one-sided $z$-test is conducted to detect the signal. We add \gls{AWGN} with a noise energy that is 900 times the background noise. The proportion of null hypotheses is $12.8\%$. 

We compare different MHT methods on the seismic signal detection task. Our proposed method is implemented as follows: We set $f_{\mathrm{mix}}(p\mid\ttgamma(v,t)) = \sigmoid\circ\ttgamma(v,t)p^{\sigmoid\circ\ttgamma(v,t) - 1}$, where $\sigmoid(x) = \ofrac{1+e^{-x}}$. We use the graph Laplacian as the \gls{GSO}, and $\set{\psi_k}$ the trigonometric basis of $L^2[-\pi, \pi]$, i.e., the set $\set{\ofrac{\sqrt{\pi}}\cos(kt)\given k\geq 1}\cup\set{\ofrac{\sqrt{\pi}}\sin(kt)\given k\geq 1}\cup\set{\ofrac{\sqrt{2\pi}}\Ind[-\pi,\pi]}$. We suppose the sample set $\sfS$ is given by $\calV\times \set{-\pi+\frac{j}{T}\cdot2\pi\given j=0,\dots, T}$. We compare our proposed method with \gls{BH} \cite{BenHoc:J95, Efr:10}, lfdr-sMoM \cite{GolZouKoi:J22}, FDR-smoothing \cite{TanKoyPolSco:J18}, SABHA \cite{LiBar:J18}, and AdaPT \cite{LeiFit:J18}. The performance is shown in \cref{fig:det_results}.

\begin{figure}[htbp]
	\centering
	\begin{subfigure}[b]{0.49\textwidth}
		\centering
		\includegraphics[width=0.7\linewidth, trim=0.55cm 0.3cm 1.5cm 1.4cm, clip]{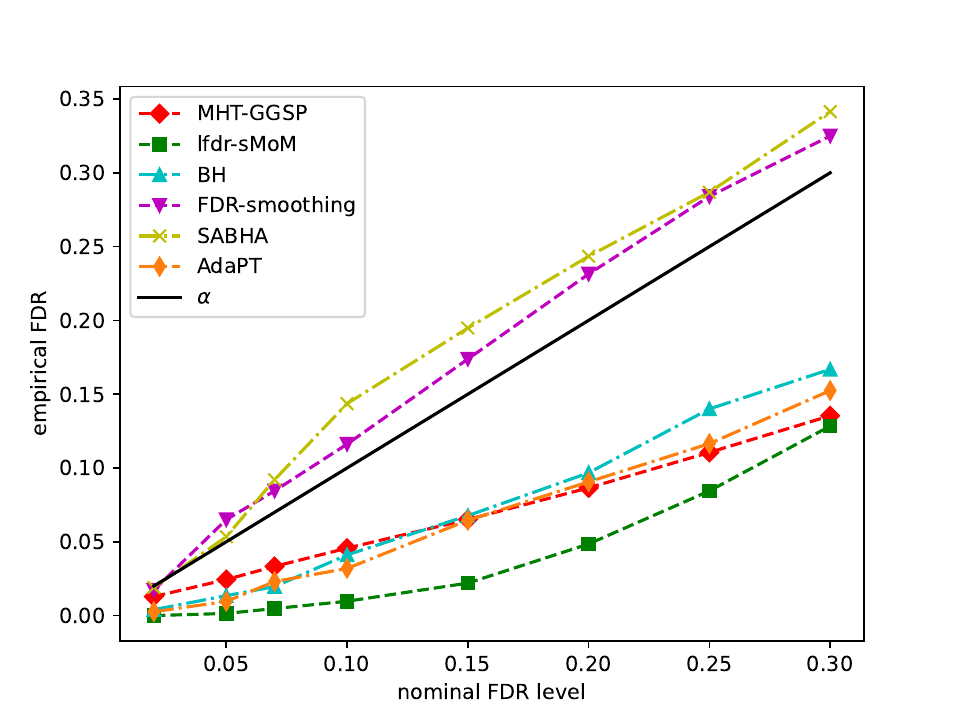}
		\caption{Empirical \gls{FDR} on seismic data}
		\label{fig:FDR_seismic}
	\end{subfigure}
	\begin{subfigure}[b]{0.49\textwidth}
		\centering
		\includegraphics[width=0.7\linewidth, trim=0.55cm 0.3cm 1.5cm 1.4cm, clip]{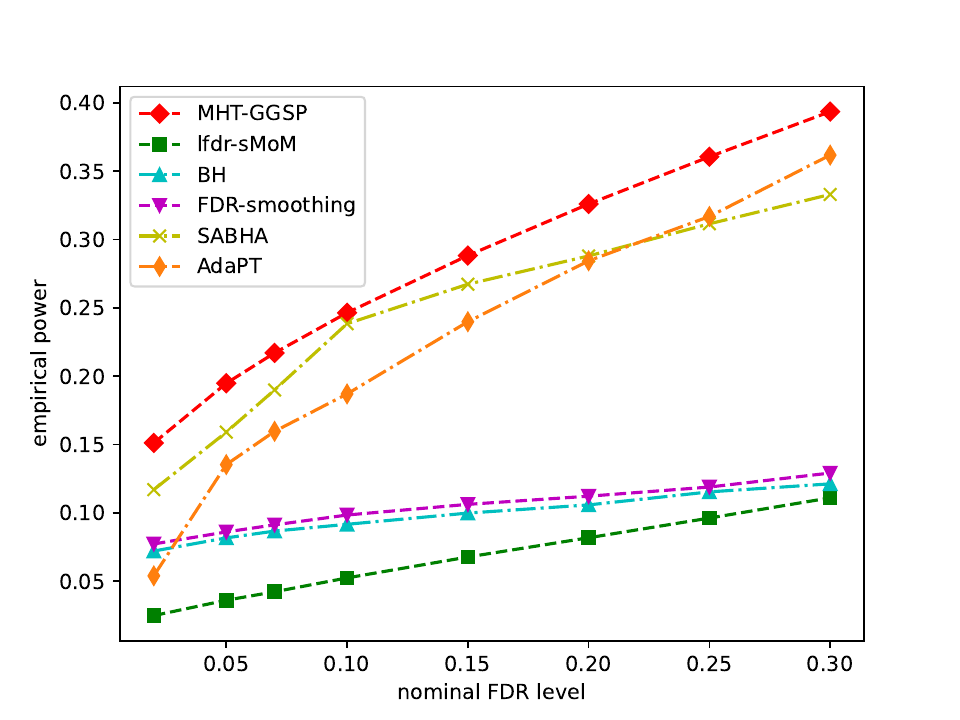}
		\caption{Empirical power on seismic data}
		\label{fig:pow_seismic}
	\end{subfigure}
	\caption{\gls{FDR} and detection power under different target \gls{FDR} levels. Each point is obtained by $20$ repetitions.}
	\label{fig:det_results}
\end{figure}

From the results, we see that FDR-smoothing and SABHA show significant \gls{FDR} inflation. We note that the FDR-smoothing assumes the same distribution of $p$-values under the alternative distribution. This assumption does not hold in general since the signal energy may be different for different vertex and time, hence $f_1(p;(v,t))$ will be different for different $(v,t)$. The SABHA method has specific constraints on the Rademacher complexity of the feasible set of reciprocal of the weights, which may not be achievable in practice, thus also observes \gls{FDR} inflation. The lfdr-sMoM method only assumes a finite number of $p$-values' distribution, hence shows a lack of detection power. Since the empirical-Bayesian model of MHT-GGSP allows the prior probability of null hypotheses and alternative distribution of $p$-values to vary on every point of $\calJ$, the resulting detection strategy observes the best power.

\bibliographystyle{IEEEtran}
\bibliography{bib/IEEEabrv,bib/StringDefinitions,bib/refs}

\end{document}